\documentclass[11pt]{article}
\usepackage{macros} 
\usepackage[mathscr]{euscript}

\renewcommand{\1}{ {\rm 1\hspace*{-0.4ex}\rule{0.1ex}{1.52ex}\hspace*{0.2ex}} }
\title{AG codes have no list-decoding friends: \\ \Large Approaching the generalized Singleton bound requires exponential alphabets\thanks{This paper was presented in part at SODA 2024.}}
\author{Omar Alrabiah\thanks{Department of EECS, UC Berkeley, Berkeley, CA, 94709, USA. Email: \url{oalrabiah@berkeley.edu}. Research supported in part by a Saudi Arabian Cultural Mission (SACM) Scholarship, NSF CCF-2210823 and V.\ Guruswami's Simons Investigator Award.} \and Venkatesan Guruswami\thanks{Departments of EECS and Mathematics, and the Simons Institute for the Theory of Computing, UC Berkeley, Berkeley, CA, 94709, USA. Email: \url{venkatg@berkeley.edu}. Research supported by a Simons Investigator Award and NSF grants CCF-2210823 and CCF-2228287.}  \and Ray Li\thanks{Department of EECS, UC Berkeley, Berkeley, CA, 94709, USA. Email: \url{rayyli@berkeley.edu}. Research supported by the NSF Mathematical Sciences Postdoctoral Research Fellowships Program under Grant DMS-2203067, and a UC Berkeley Initiative for Computational Transformation award.}}
\mathtoolsset{showonlyrefs}
\begin{document}

\maketitle
\thispagestyle{empty}

\begin{abstract}
A simple, recently observed generalization of the classical Singleton bound to list-decoding asserts that rate $R$ codes are not list-decodable using list-size $L$ beyond an error fraction $\tfrac{L}{L+1} (1-R)$ (the Singleton bound being the case of $L=1$, i.e., unique decoding). We prove that in order to approach this bound for any fixed $L >1$, one needs exponential alphabets. Specifically, for every $L>1$ and $R\in(0,1)$, if a rate $R$ code can be list-of-$L$ decoded up to error fraction $\tfrac{L}{L+1} (1-R -\eps)$, then its alphabet must have size at least $\exp(\Omega_{L,R}(1/\eps))$. This is in sharp contrast to the situation for unique decoding where certain families of rate $R$ algebraic-geometry (AG) codes over an alphabet of size $O(1/\varepsilon^2)$ are unique-decodable up to error fraction  $(1-R-\varepsilon)/2$. 
Our bounds hold even for subconstant $\varepsilon\ge 1/n$, implying that any code exactly achieving the $L$-th generalized Singleton bound requires alphabet size $2^{\Omega_{L,R}(n)}$. Previously this was only known only for $L=2$ under the additional assumptions that the code is both linear and MDS.

\smallskip Our lower bound is tight up to constant factors in the exponent---with high probability random codes (or, as shown recently, even random linear codes) over $\exp(O_L(1/\varepsilon))$-sized alphabets, can be list-of-$L$ decoded up to error fraction $\tfrac{L}{L+1} (1-R -\varepsilon)$. 
\end{abstract}

\section{Introduction}

The Singleton bound \cite{Singleton64} states that a code of rate $R$ cannot uniquely correct a fraction of worst-case errors exceeding $\frac{1}{2}(1-R)$. The straightforward generalization of this bound to list decoding implies that one cannot do list-of-$L$ decoding (where the decoder must output at most $L$ codewords) beyond an error fraction of $\tfrac{L}{L+1} (1-R)$. See Figure~\ref{fig:singleton} for an illustration of this bound, which has been called the \emph{generalized Singleton bound}~\cite{ST20}.

\begin{figure}
  \begin{center}
    \newcommand\dashvert[2]{\draw[dashed] (#1,-1.5) -- (#1,#2);}
    \newcommand\tick[1]{\draw (#1,-1.6) -- (#1,-1.4); \dashvert{#1}{2.7}} 
    \begin{tikzpicture}[scale=0.6]
        \def\ep{2}
        \def\ii{4}
        \def\ca{gray}
        \def\cb{yellow!40}
        \def\cc{orange!40}
        \def\ce{red!40}
        \def\cd{cyan}

        \draw[fill=\ce] (-\ep,2) rectangle (4*\ii,2.5);
        \draw[fill=\cc] (-\ep,1) rectangle (4*\ii,1.5);
        \draw[fill=\cb] (-\ep,0) rectangle (4*\ii,0.5);
        \draw[fill=\ca] (-\ep,-1) rectangle (4*\ii,-.5);
        
        \draw[fill=\cd] (-\ep,2) rectangle (\ii,2.5);
        \draw[fill=\cd] (-\ep,1) rectangle (\ii,1.5);
        \draw[fill=\cd] (-\ep,0) rectangle (\ii,0.5);
        \draw[fill=\cd] (-\ep,-1) rectangle (\ii,-.5);
        \draw[fill=\ce] (\ii,-1) rectangle (2*\ii,-0.5);
        \draw[fill=\cc] (2*\ii,-1) rectangle (3*\ii,-0.5);
        \draw[fill=\cb] (3*\ii,-1) rectangle (4*\ii,-0.5);

        \tick{-\ep}
        \tick{\ii}
        \tick{2*\ii}
        \tick{3*\ii}
        \tick{4*\ii}
        \draw (-\ep,-1.5) -- (4*\ii,-1.5);
        \draw[line width=3pt] (-\ep,-1) rectangle (4*\ii,-.5);
        
        \node at (-0.5-\ep,2.25) {$c_0$};
        \node at (-0.5-\ep,0.25) {$c_2$};
        \node at (-0.5-\ep,1.25) {$c_1$};
        \node at (-0.5-\ep,-0.75) {$y$};
        \node at ({0.5*(-\ep+\ii)},-2) {$Rn$};
        \node at (1.5*\ii,-2) {$\frac{1}{3}(1-R)n$};
        \node at (2.5*\ii,-2) {$\frac{1}{3}(1-R)n$};
        \node at (3.5*\ii,-2) {$\frac{1}{3}(1-R)n$};
  \end{tikzpicture}
  \end{center}
  \caption{The generalized Singleton bound, illustrated for $L=2$. In any code of rate $R$, by pigeonhole, there are three codewords $c_0,c_1,c_2$ that agree on the first $Rn-O(1)$ coordinates. Then there is a ``list-decoding center'' $y$ that differs from each of $c_0, c_1,c_2$ on at most $\frac{2}{3}(1-R)n +O(1)$ coordinates, so the code is not $(\frac{2}{3}(1-R)+o(1),2)$-list-decodable.}
  \label{fig:singleton} 
\end{figure}
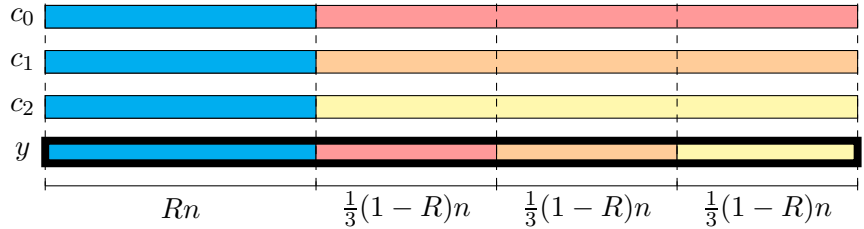
Our main result is that approaching the generalized Singleton bound within $\varepsilon$ requires an alphabet size exponential in $1/\varepsilon$. We say a code $C \subset \Sigma^n$ is $(p,L)$-list decodable if for every $y \in \Sigma^n$, there are at most $L$ codewords of $C$ within Hamming distance $p n$ from $y$. Formally, we prove:

\begin{theorem}
  Let $L\ge2$ be a fixed constant and $R\in(0,1)$.
  There exists an absolute constant $\alpha_{L,R}$ such that the following holds for all $\varepsilon>0$ and all sufficiently large $n\ge\Omega_{L,R}(1/\varepsilon)$. 
  Let $C$ be a code of length $n$ with alphabet size $q$ that is $(\frac{L}{L+1}(1-R-\varepsilon),L)$-list-decodable.
  Then, $q\ge 2^{\alpha_{L,R}/\varepsilon}$.
  \label{thm:main}
\end{theorem}

The exponential alphabet size lower bound in Theorem~\ref{thm:main} is in sharp contrast with the situation for unique decoding (the $L=1$ case), where, for any desired rate $R$, certain families of algebraic-geometric (AG) codes over an alphabet of size $O(1/\eps^2)$ allow unique decoding up to an error fraction $(1-R-\eps)/2$ with rate $R$ \cite{TVZ82, GS95}.

The Plotkin bound \cite{Plotkin60} implies a lower bound of $\Omega(1/\eps)$ for such unique decodability,\footnote{The Plotkin bound is typically stated as follows: a code with relative distance  at least $\frac{q-1}{q}$ has size $O(n)$. This alphabet dependent version follows by, in a rate $R$ code, pigeonholing to find $O(n)$ codewords agreeing on the first roughly $R n$ coordinates, and finding (by the Plotkin bound) two codewords that additionally agree on $\frac{1}{q}$ fraction of the remaining coordinates, so that we need $q\ge \Omega(1/\varepsilon)$.} and AG codes come within a quadratic factor of this bound. However, for list decoding with any fixed list-size $L >1$, there is no such AG-like polynomial convergence (as a function of alphabet size) to the optimal trade-off $\tfrac{L}{L+1}(1-R)$. In fact, the convergence is exponentially slow.

\begin{remark}
\label{remark:growing-l}
    For a code $C \subset [q]^n$ of rate $R$, note that a random Hamming ball of radius $p n$ has in expectation $q^{h_q(p) n -o(n)} q^{(R-1)n}$ codewords.\footnote{Here, $h_q(x)\defeq x\log_q(q-1) - x\log_qx - (1-x)\log_q(1-x)$ is the $q$-ary entropy function.} For $C$ to be $(p,L)$-list decodable one must therefore have $h_q(p) \le 1-R +o(1)$. (This trade-off is the list-decoding capacity for codes of alphabet size $q$.) A straightforward computation then implies a lower bound of $q\ge 2^{\Omega_R(\min(L,1/\varepsilon))}$ on the alphabet size of a family of $(\frac{L}{L+1}(1-R-\varepsilon),L)$-list-decodable codes.
For a fixed $L$, this lower bound does not scale with $\eps$.
In comparison, we get an exponential in $1/\eps$ lower bound for any fixed list size $L$.
\end{remark}

Our work builds on the recent work of 
Brakensiek, Dhar, and Gopi~\cite{BDG22} who proved the following result for MDS codes that are list-decodable all the way up to the generalized Singleton bound. Recall that a (linear) MDS code is one whose dimension $k$, minimum distance $d$, and block length $n$ satisfy $k+d=n+1$.
\begin{theorem}[\cite{BDG22}]
  Let $R\in(0,1)$.
  Any linear MDS code that is $(\frac{2}{3}(1-R),2)$-list-decodable must have alphabet size at least $2^{\Omega_R(n)}$.
  \label{thm:bdg}
\end{theorem}

We present a comparison of our Theorem~\ref{thm:main} to Theorem~\ref{thm:bdg} in Section~\ref{ssec:comparison-BDG}. To summarize this comparison, our result generalizes Theorem~\ref{thm:bdg} in four ways: our result (i) applies to general (not necessarily linear) codes, (ii) incorporates the gap to capacity $\varepsilon$ (iii) removes the MDS assumption (or more generally any assumption on the code distance), and (iv) generalizes to larger $L$. In essence, our proof distills the proof of Theorem~\ref{thm:bdg} to its combinatorial core, and then adds new ideas to enable these generalizations.

We note that Theorem~\ref{thm:main} strictly generalizes Theorem~\ref{thm:bdg} (with a worse constant in the exponent), since we can take $\varepsilon=\Theta(1/n)$. Setting, $\varepsilon=\Theta(1/n)$ we have that any $(\frac{L}{L+1}(1-R),L)$-list-decodable code is certainly $(\frac{L}{L+1}(1-R-\varepsilon),L)$ and thus, by Theorem~\ref{thm:main}, must have alphabet size $2^{\Omega_{L,R}(n)}$. Thus, we prove that \emph{any} code --- not just linear MDS codes --- exactly achieving the generalized Singleton bound must have $2^{\Omega_{L,R}(n)}$-sized alphabets.

One can show that an exponential in $1/\eps$ alphabet size suffices to approach the generalized Singleton bound within $\eps$. The argument is a simple random coding argument with alterations (see, e.g., \cite[Chapter 3]{AS16}), which we present for completeness in Appendix~\ref{sec:codes-exist}. 
\begin{proposition}
\label{pr:codes-exist}
  Let $L\ge1$ be a fixed constant, let $R\in(0,1)$, and let $\varepsilon\in(0,1)$.
  There exists a code $C$ over an alphabet size $q\le 2^{O(1/\varepsilon)}$ that is $(\frac{L}{L+1}(1-R-\varepsilon),L)$-list-decodable.
\end{proposition}

Thus the lower bound in Theorem~\ref{thm:main} is tight up to the constant factor $\alpha_{L,R}$ in the exponent (which we did not try to optimize).
This is perhaps surprising, because Proposition~\ref{pr:codes-exist}  considers the natural random construction, which does \emph{not} give near-optimal alphabet size for unique decoding ($L=1$). In fact, for unique decoding, the alphabet size exponentially far from optimal: $2^{O(1/\varepsilon)}$ compared to the optimal $\poly(1/\varepsilon)$ achieved by AG codes.
  
  We also point out that, as a consequence of recent work \cite{AGL23} on randomly punctured codes (building on \cite{GZ23}), Proposition~\ref{pr:codes-exist} can be achieved not just with random codes (with alterations), but also with \emph{random linear codes}. \cite{AGL23} shows that, for all $L\ge 2$ and $R,\varepsilon\in(0,1)$, random linear codes over alphabet size $2^{10L/\varepsilon}$ are with high probability $(\frac{L}{L+1}(1-R-\varepsilon),L)$-list-decodable.

\section{Preliminaries}
\label{sec:prelims}

We use standard Landau notation $O(\cdot),\Omega(\cdot), \Theta(\cdot)$. We use the notations $O_x(\cdot),\Omega_x(\cdot), \Theta_x(\cdot)$ to mean that a multiplicative factor depending on the variable(s) $x$ is suppressed.
All logs are base-2 unless otherwise specified.
For integers $L$, let $[L]$ denote the set $\{1,\dots,L\}$.
Let $h(x) \defeq -x\log x - (1-x)\log(1-x)$ denote the binary entropy function.
Let $\binom{n}{\le r}\defeq \sum_{i=0}^{r} \binom{n}{i}$.
We have the binomial approximation
\begin{align}
  2^{h(\alpha)n-o(n)}{\poly n}\le \binom{n}{\alpha n}\le \binom{n}{\le\alpha n} \le 2^{h(\alpha)n}
  \label{eq:binomial}
\end{align}
We also use the Chernoff bound for binomials
\begin{align}
  \Pr[\Binomial(\alpha,m) > (1+\delta)\alpha m] \le e^{-\frac{\delta^2}{2+\delta}\alpha m}.
  \label{eq:chernoff}
\end{align}

A \emph{code} $C \subseteq \Sigma^n$ is simply a collection of words of equal length over a fixed alphabet $\Sigma$. The \emph{dimension} of a code $C$ is defined to be $k(C) \defeq \log_{\abs{\Sigma}}{\abs{C}}$ and the \emph{rate} is $R(C)\defeq \frac{k(C)}{n} = \frac{\log_{\abs{\Sigma}}{\abs{C}}}{n}$. The \emph{distance} of a code $C$ is defined to be $d(C) \defeq \min_{c_1 \neq c_2 \in C}{d(c_1, c_2)}$, where $d(\cdot,\cdot)$ denotes the Hamming distance between two words.
We say a code of length $n$ and dimension $k$ is \emph{Maximum Distance Separable (MDS)} if it has minimum distance $n-k+1$.

For a string $c\in \Sigma^n$ and a set $A\subset [n]$, we let $c|_A\in\Sigma^{|A|}$ denote the string $c$ restricted to the indices in $A$.

To highlight the main ideas in Section~\ref{sec:warmups}, we prove special cases of Theorem~\ref{thm:main} for \emph{average-radius-list-decoding}: a code $C\subset[q]^n$ is \emph{$(p,L)$-average-radius-list-decodable} if, for any distinct $L+1$ codewords $c\ind{1},\dots,c\ind{L+1}$ and any vector $y\in[q]^n$, the average Hamming distance from $c\ind{1},\dots,c\ind{L+1}$ to $y$ is strictly greater than $pn$.
We observe that average-radius-list-decoding is a strengthening of (ordinary) list-decoding: any $(p,L)$-average-radius-list-decodable  code is also $(p,L)$-list-decodable.

\section{Technical overview}
\label{sec:warmups}

We now introduce the ideas of our main result, Theorem~\ref{thm:main}.
Our main result generalizes Theorem~\ref{thm:bdg} in four ways, which we outline in Section~\ref{ssec:comparison-BDG}. 
In Section~\ref{ssec:warmup-1}, Section~\ref{ssec:warmup-2}, and Section~\ref{ssec:warmup-3}, we give a few warmup proofs that show (or give a taste of) how we achieve these generalizations.
For exposition, we focus our warmup lower bounds on list size $L=2$ and for average-radius-list-decoding.
In Section~\ref{sec:proof}, we give the full proof of our main result, Theorem~\ref{thm:main}.

\subsection{Comparison to \cite{BDG22}}
\label{ssec:comparison-BDG}
Our work generalizes Theorem~\ref{thm:bdg} in four ways. 
\begin{enumerate}
  \item \textbf{Removing linearity.}
  The proof of Theorem~\ref{thm:bdg} uses that the codes in question have a ``higher order MDS'' (MDS(3)) property, which is a specific property of the columns of the parity-check matrix \cite{Roth22, BGM22}.
  They then show that a small alphabet size contradicts the higher order MDS property, and hence the assumption of the code in Theorem~\ref{thm:bdg}.
  Higher order MDS codes can only be defined for linear codes, and furthermore, the proof of \cite{BDG22} used several aspects of the linearity of the code.
  Thus, on the surface, it may seem like Theorem~\ref{thm:bdg} could not be generalized to non-linear codes.

We show that the linearity assumption is in fact not necessary. 
We show that one can avoid the linear-algebraic aspects of the proof in \cite{BDG22}, and that careful applications of the pigeonhole principle suffice to find a bad list-decoding configuration. 
In our first warmup (Section~\ref{ssec:warmup-1}), we show how to do this (for average-radius-list-decoding).

  \item \textbf{Allowing gap to capacity.}
  While Theorem~\ref{thm:bdg} only proves a lower bound for codes list-decodable exactly up to the generalized Singleton bound, we prove a lower bound even when the code has an $\varepsilon$ gap to capacity, showing an alphabet size lower bound for codes approaching list-of-$L$ capacity, for any $L$.
    In our second warmup (Section~\ref{ssec:warmup-2}), we show how to do this (for average-radius-list-decoding).

  \item \textbf{Removing MDS.}
  In the connection between list-decodable codes and higher order MDS codes, a code is ``MDS($L+1$)'' if and only if it exactly achieves the generalized Singleton bound for all $L'\le L$ \cite{BGM23}.
    The lower bound in Theorem~\ref{thm:bdg} is proved for MDS(3), which requires Theorem~\ref{thm:bdg} to assume our code both (i) meets the generalized Singleton bound for $L=2$ \emph{and} (ii) is MDS.
    Using arguments similar to the Johnson-bound we show that one can get away with a weaker distance assumption than MDS, and by adjusting our pigeonhole argument, we then can eliminate the assumption entirely. 
    We give a taste of how to do this in our third warmup by showing how to remove the MDS/distance assumption for average-radius-list-decoding (Section~\ref{ssec:warmup-3}) --- it is only a taste, as removing the distance assumption is much easier for average-radius-list-decoding than for ordinary list-decoding.

  \item \textbf{Generalizing to larger $L$.}
  In contrast to lower bounds for higher order MDS codes, generalizing the list-decoding lower bounds to larger $L$ is not immediate. \cite{BDG22} proved an alphabet size lower bound of $q\ge 2^{\Omega_R(n)}$ for MDS(3) codes. Since all MDS(3) codes are also MDS($L$) for $L\ge 3$, their lower bound $q\ge 2^{\Omega_R(n)}$ also held for MDS($L$) for all $L\ge 3$.
    However, while the lower bounds for $L=2$ imply lower bounds for larger $L$ in higher order MDS codes, the same is \emph{not} true for lower bounds for list-of-$L$ decoding: optimal list-of-$L$ decoding does not necessarily imply optimal list-of-$L'$ decoding for $L' > L$.

    Generalizing all the above machinery to larger $L$ for average-radius list-decoding follows almost seamlessly from the warmup arguments. However, the list-decoding case is not as immediate. It requires new ideas to remove the distance assumption, care to find the bad list-decoding configuration, and a deliberate balancing of parameters to ensure that all distances from the codewords to the center are simultaneously below the list-decoding radius. 
\end{enumerate}

We now give several warmup proofs which show how to achieve these generalizations.
Warmup 1 (Section~\ref{ssec:warmup-1}) achieves the first generalization (removing linearity), Warmup 2 (Section~\ref{ssec:warmup-2}) achieves the second generalization (incorporating gap to capacity), and Warmup 3 (Section~\ref{ssec:warmup-3}) achieves the third generalization (removing MDS), all for the easier case of average-radius-list-decoding and $L=2$.
The full proof of Theorem~\ref{thm:main} incorporates all of these ideas simultaneously and additionally achieves the fourth generalization (all $L$).

\subsection{Warmup 1: A lower bound for exactly optimal list-of-2 decoding}
\label{ssec:warmup-1}

First, we show how to prove a lower bound for all (not-necessarily-linear) codes.
In other words, we generalize Theorem~\ref{thm:bdg} to a lower bound for all codes. 
We state and prove the lower bound for average-radius-list-decoding, though, as demonstrated by our main result, Theorem~\ref{thm:main}, a similar lower bound holds for ordinary list-decoding.\footnote{Our main result, Theorem~\ref{thm:main}, requires $\varepsilon\ge \Omega(1/n)$, but for such $\varepsilon=\Theta(1/n)$, a $(\frac{2}{3}(1-R),2)$-list-decodable code is certainly $(\frac{2}{3}(1-R-\varepsilon),2)$-list-decodable, so Theorem~\ref{thm:main} implies $q\ge 2^{\Omega_{L,R}(n)}$.}
\begin{proposition}
  For all $R\in(0,1)$, there exists $\alpha_R>0$ such that the following holds for sufficiently large $n$.
  Any MDS code that is $(\frac{2}{3}(1-R),2)$-average-radius-list-decodable must have alphabet size $q\ge 2^{\alpha_R \cdot n}$.
  \label{pr:main-0}
\end{proposition}
\begin{proof}
  Fix $I_0=\{1,2\}$.
  Let $\mathcal{F}$ be the collection of all subsets of $[n]\setminus I_0$ of size $k- 1$.
  Clearly $\abs{\mathcal{F}} \ge 2^{\Omega_R(n)}$.
  Thus, it suffices to prove that $q^2 \ge \abs{\mathcal{F}}/2$. 
  Suppose for contradiction that
  \begin{align}
    q^{2} < |\mathcal{F}|/2.
    \label{eq:false-0}
  \end{align}
  Consider picking a uniformly random codeword $c \in C$. 
  For each $A \in \mathcal{F}$, let $\mathscr{E}_A$ be the event that another codeword $c'$ agrees with $c$ on $A$, i.e., $c|_A = c'|_A$. For any $A \in \mathcal{F}$, there at most $q^{k-1}$ possible values of $c|_A$, and thus at most $q^{k-1}$ codewords $c$ for which $c|_A$ uniquely determines $c$. Hence,
  \begin{equation}
      \Pr[\neg\mathscr{E}_A] < \frac{q^{k-1}}{q^k} = \frac{1}{q} \ .
      \label{eq:uniqueness-ub-0}
  \end{equation}
 For each codeword $c$, define the set $\mathcal{F}_c \coloneqq \{A \in \mathcal{F} : \text{$\mathscr{E}_A$ occurs}\}$. For each $A \in \mathcal{F}_c$, we can find, by definition, a codeword $f^A(c) \in C \setminus \{c\}$ such that $f^A(c)|_A = c|_A$. 
  By linearity of expectation and \eqref{eq:uniqueness-ub-0}, we find that
  \begin{equation}
      \E\left[\abs{\mathcal{F}_c}\right] = \E\left[\sum_{A \in \mathcal{F}}{\1\{\mathscr{E}_A\}}\right] > \sum_{A \in \mathcal{F}}{\left(1 - \frac{1}{q}\right)} \ge \frac{\abs{\mathcal{F}}}{2} \ .
  \end{equation}
  Hence we can find a codeword $c \in C$ for which $\abs{\mathcal{F}_c} > \abs{\mathcal{F}}/2$. By pigeonhole and \eqref{eq:false-0}, there are $2$ distinct sets $A_1,A_2 \in \mathcal{F}_c$ such that the codewords $f^{A_1}(c)$ and $f^{A_2}(c)$ agree on the coordinates $I_0$. 
  These two codewords $f^{A_1}(c)$ and $f^{A_2}(c)$ are distinct: if not, then $f^{A_1}(c)=f^{A_2}(c)$ agrees with $c$ on at least $\abs{A_1 \cup A_2} \ge k$ coordinates, contradicting the assumption that the distance is at least $n-k+1$.

  Let $y$ be the word which agrees with $f^{A_1}(c)$ on $I_0$ (and thus $f^{A_2}(c)$ as well), and agrees with $c$ everywhere else.
  Word $y$ has total distance at most
  \begin{equation}
      \abs{I_0}+ \abs{[n] \setminus (I_0 \cup A_1)} + \abs{[n] \setminus (I_0 \cup A_2)} = 2 + (n-k-1) + (n-k-1) = 2(n-k)
      \label{eq:avg-dist-0}
  \end{equation}
  from codewords $c,f^{A_1}(c),f^{A_2}(c)$, contradicting average-radius-list-decoding.
  Thus, \eqref{eq:false-1} is false, which means $2q^2\ge |\mathcal{F}|\ge 2^{\Omega_R(n)}$, and so $q\ge 2^{\Omega_R(n)}$.
\end{proof}

\subsection{Warmup 2: Relaxing by $\varepsilon$}
\label{ssec:warmup-2}

\begin{figure}
  \begin{center}
    \newcommand\dashvert[2]{\draw[dashed] (#1,-1.5) -- (#1,#2);}
    \newcommand\tick[1]{\draw (#1,-1.6) -- (#1,-1.4); \dashvert{#1}{1.5}} 
    \begin{tikzpicture}[scale=0.8]
        \def\ep{0.5}
        \def\ii{3}
        \def\ca{gray}
        \def\cb{yellow!40}
        \def\cc{orange!40}
        \def\cd{cyan}
        \def\ce{red!40}
          
        \draw[fill=\ce] (-\ep,2) rectangle (16,2.5);
        \draw[fill=\ca] (-\ep,1) rectangle (16,1.5);
        \draw[fill=\ca] (-\ep,0) rectangle (16,0.5);
        \draw[fill=\ce] (-\ep,-1) rectangle (16,-.5);    
        
        \draw[fill=\cd] (-\ep,1) rectangle (0,1.5);
        \draw[fill=\cd] (-\ep,0) rectangle (0,0.5);
        \draw[fill=\cd] (-\ep,-1) rectangle (0,-.5);
        \draw[fill=\ce] (1.5*\ii,1) rectangle (3.5*\ii,1.5);
        \draw[fill=\ce] (2.9*\ii,0) rectangle (4.9*\ii,0.5);
        
        \tick{-\ep}
        \tick{0}
        \tick{1.5*\ii}
        \tick{3.5*\ii}
        \draw (-\ep,-1.5) -- (16,-1.5);
        \draw[line width=3pt] (-\ep,-1) rectangle (16,-.5);

        \node at (-3*\ep,2.25) {$c$};
        \node at (-3*\ep,1.25) {$f^{A_1}(c)$};
        \node at (-3*\ep,0.25) {$f^{A_2}(c)$};
        \node at (-3*\ep,-0.75) {$y$};
        \node at (-0.5*\ep,3) {$I_0$};
        \node at (-0.5*\ep,-2) {$4\varepsilon n$};
        \node at (2.5*\ii,-2) {$(R-\varepsilon)n$};
        \node at (2.5*\ii,1.25) {$A_1$};
        \node at (3.9*\ii,0.25) {$A_2$};
  \end{tikzpicture}
  \end{center}
  \caption{The bad average-radius-list-decoding configuration we search for in Proposition~\ref{pr:main-1}. The list-decoding center $y$ has distances $4\varepsilon n$, $(1-R-3\varepsilon)n$, and $(1-R-3\varepsilon)n$ from codewords $c_0$, $c_1$, and $c_2$ respectively. 
  }
  \label{fig:agreement-0} 
\end{figure}
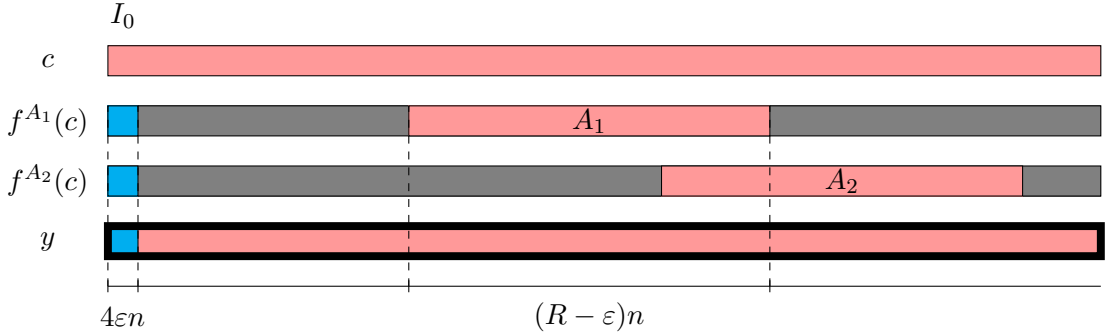

Next, we show how to prove an alphabet size lower bound of $2^{\Omega_R(1/\varepsilon)}$ (Proposition~\ref{pr:main-1}), assuming the code has a gap-to-capacity of $\varepsilon$, for average-radius-list-of-2 decoding. We additionally assume the code has near-optimal minimum distance, similar to Proposition~\ref{pr:main-0} and Theorem~\ref{thm:bdg}, and then show how one can remove it in Section~\ref{ssec:warmup-3}. 

The proof of Proposition~\ref{pr:main-1} follows nearly the same blueprint as the proof of Proposition~\ref{pr:main-0}. The new addition will be increasing the size of $I_0$ to be $\Omega(\varepsilon n)$. That way, the bound in~\eqref{eq:avg-dist-0} decreases by a factor of $2\varepsilon n$ so that it still contradicts average-radius-list-decodability. In exchange, the bound of $2q^2 \le \abs{\mathcal{F}}$ is altered to become $2q^{\abs{I_0}} \le \abs{\mathcal{F}}$.

\begin{proposition}
  For all $R\in(0,1)$, there exists $\alpha_R>0$ such that the following holds for all sufficiently large $n$ and all $\varepsilon\ge 1/n$.
  Let $C$ be a code of rate $R$ with alphabet size $q$ that has minimum distance greater than $(1-R-\varepsilon) n$ and is $(\frac{2}{3}(1-R-\varepsilon),2)$-average-radius-list-decodable. Then $q\ge 2^{\alpha_R/\varepsilon}$.
  \label{pr:main-1}
\end{proposition}
\begin{proof}
  Since we allow $\varepsilon\ge 1/n$, we need to deal with rounding issues. 
  It suffices to consider the case where $R$ and $\varepsilon$ are both integer multiples of $3/n$: if the Proposition is true in these cases with constant $\alpha_R'$, and we are give a code $C$ of rate $R$ and satisfying the Proposition's conditions for some $\varepsilon$, we can round $R$ down to the nearest multiple of $3/n$, call it  $R'$,  and round $\varepsilon+3/n$ up to the next multiple of $3/n$, call it $\varepsilon'$. 
  Taking $C'$ to be any subcode of $C$ of rate $R'$, we can check $C'$ and $\varepsilon'$ satisfy the requirements of the Proposition (as $1-R-\varepsilon > 1-R'-\varepsilon'$), so $q\ge 2^{\alpha'_R/\varepsilon'}\ge 2^{\alpha_R/\varepsilon}$ where $\alpha_R = \alpha'_R/10$.
  Thus, we assume for the rest of the proof $R$ and $\varepsilon$ are multiples of $3/n$ (so that everything that ``should'' be an integer is an integer).

By adjusting $\alpha_R$, it suffices to consider $\varepsilon$ sufficiently small compared to $R$.
  Fix $I_0=\{1,2,\dots,4\varepsilon n\}$.
  Let $\alpha \coloneqq R-\varepsilon$ and $\beta \coloneqq R+\varepsilon$.
  For any two subsets $A, B \subseteq [n] \setminus I_0$ satisfying $\abs{A} = \abs{B} = \alpha n$ and $\abs{A \cup B} \le \beta n$, notice that $\abs{A \setminus B} = \abs{B \setminus A} \le (\beta - \alpha)n = 2\varepsilon n$. Since the tuple of sets $(A, A \setminus B, B \setminus A)$ determine $B$, the number of possible subsets $B$ for any given $A$ is therefore at most $\binom{(1-4\varepsilon)n}{\le 2\varepsilon n}^2$. Thus, by greedily choosing subsets, we can find a family $\mathcal{F}$ of $\binom{(1-4\varepsilon)n}{\alpha n} / \binom{(1-4\varepsilon)n}{\le 2\varepsilon n}^2 \ge 2^{\Omega_R(n)}$ subsets of $[n] \setminus I_0$ such that each subset has size $\alpha n$ and any pairwise union has size at least $\beta n$.

  It suffices to prove that $q^{|I_0|}\ge |\mathcal{F}|/2$.
  Suppose for contradiction that
  \begin{align}
    q^{|I_0|} < \abs{\mathcal{F}} / 2.
    \label{eq:false-1}
  \end{align}
  Consider picking a uniformly random codeword $c \in C$. 
  For each $A \in \mathcal{F}$, let $\mathscr{E}_A$ be the event that another codeword $c'$ agrees with $c$ on $A$, i.e., $c|_A = c'|_A$. For any $A \in \mathcal{F}$, there at most $q^{\alpha n}$ possible values of $c|_A$, and thus at most $q^{\alpha n}$ codewords $c$ for which $c|_A$ uniquely determines $c$. Hence,
  \begin{equation}
      \Pr[\neg\mathscr{E}_A] < \frac{q^{\alpha n}}{q^k} = \frac{1}{q^{\varepsilon n}} \ .
      \label{eq:uniqueness-ub-1}
  \end{equation}
For each codeword $c$, define the set $\mathcal{F}_c \coloneqq \{A \in \mathcal{F} : \text{$\mathscr{E}_A$ occurs}\}$. For each $A \in \mathcal{F}_c$, we can find, by definition, a codeword $f^A(c) \in C \setminus \{c\}$ such that $f^A(c)|_A = c|_A$.
  By linearity of expectation and \eqref{eq:uniqueness-ub-1}, we thus find that
  \begin{equation}
      \E\left[\abs{\mathcal{F}_c}\right] = \E\left[\sum_{A \in \mathcal{F}}{\1\{\mathscr{E}_A\}}\right] > \sum_{A \in \mathcal{F}}{\left(1 - \frac{1}{q^{\varepsilon n}}\right)} \ge \frac{\abs{\mathcal{F}}}{2} \ .
  \end{equation}
  Hence we can find a codeword $c \in C$ for which $\abs{\mathcal{F}_c} > \abs{\mathcal{F}}/2$. By pigeonhole and \eqref{eq:false-1}, there are $2$ distinct sets $A_1,A_2 \in \mathcal{F}_c$ such that the codewords $f^{A_1}(c)$ and $f^{A_2}(c)$ agree on the coordinates $I_0$. 
  These two codewords $f^{A_1}(c)$ and $f^{A_2}(c)$ are distinct: if not, then $f^{A_1}(c)=f^{A_2}(c)$ agrees with $c$ on at least $\abs{A_1 \cup A_2} \ge \beta n = (R + \varepsilon)n$ coordinates, contradicting the assumption that the distance is greater than $(1-R-\varepsilon)n$.

  Let $y$ be the word which agrees with $f^{A_1}(c)$ on $I_0$ (and thus $f^{A_2}(c)$ as well), and agrees with $c$ everywhere else (see Figure~\ref{fig:agreement-0}).
  Word $y$ has total distance at most
  \begin{align}
      \abs{I_0}+ \abs{[n] \setminus (I_0 \cup A_1)} + \abs{[n] \setminus (I_0 \cup A_2)} &= 4\varepsilon n + (1-R-3\varepsilon)n + (1-R-3\varepsilon)n \nonumber \\
      &= 2(1-R-\varepsilon)n
      \label{eq:avg-dist-1}
  \end{align}
  from codewords $c,f^{A_1}(c),f^{A_2}(c)$, contradicting average-radius-list-decoding.
  Thus, \eqref{eq:false-1} is false, which means $2q^{4\varepsilon n} \ge |\mathcal{F}|\ge 2^{\Omega_R(n)}$, and so $q\ge 2^{\Omega_R(1/\varepsilon)}$.
\end{proof}

\subsection{Warmup 3: Removing the distance assumption}
\label{ssec:warmup-3}

In this section, we prove Proposition~\ref{pr:main-2}, which is the same as Proposition~\ref{pr:main-1} but with the distance assumption removed. To remove it, we simply observe that the minimum distance is already nearly satisfied in any $(\frac{2}{3}(1-R-\varepsilon),2)$-average-radius-list-decodable code. This is \emph{not} true for ordinary list-decoding, so we need additional ideas to remove the minimum distance condition in the general lower bound, Theorem~\ref{thm:main}, but we include this much simplier proof to illustrate the high level structure of the proof. 

\begin{proposition}
  For all $R\in(0,1)$, there exists $\alpha_R>0$ such that the following holds for all $\varepsilon\in(0,1)$ and all sufficiently large $n\ge\Omega_R(1/\varepsilon)$.
  Let $C$ be a code of rate $R$ with alphabet size $q$ that is $(\frac{2}{3}(1-R-\varepsilon),2)$-average-radius-list-decodable. Then $q\ge 2^{\alpha_R/\varepsilon}$.
  \label{pr:main-2}
\end{proposition}

To prove Proposition~\ref{pr:main-2}, we need the following simple lemma.

\begin{lemma}
  In any $(p,2)$-average-radius-list-decodable code, each codeword has at most 1 other codeword within distance $\frac{3p}{2} n$.
  \label{lem:distance-0}
\end{lemma}
\begin{proof}
  If there are 2 codewords $c_1$ and $c_2$ both within distance $3pn/2$ of codeword $c$, then the word $c$ has average distance $pn$ from the codewords $c,c_1,c_2$, contradicting $(p,2)$-average-radius-list-decodability.
\end{proof}

Now we can prove Proposition~\ref{pr:main-2}.
\begin{proof}[Proof of Proposition~\ref{pr:main-2}]
  Let $C$ be a $(\frac{2}{3}(1-R-\varepsilon),2)$-average-radius-list-decodable code.
  By Lemma~\ref{lem:distance-0}, each codeword has at most 1 other codeword within distance $(1-R-\varepsilon)n$.
  Thus, by choosing codewords greedily, $C$ has a subcode $C'$ of size at least $|C|/2$ that both is $(\frac{2}{3}(1-R-\varepsilon),2)$-average-radius-list-decodable and has minimum distance greater than $(1-R-\varepsilon)n$.
  Subcode $C'$ has rate at least $R'=R-(1/n)$.
  Applying Proposition~\ref{pr:main-1} with subcode $C'$, rate $R'=R-(1/n)$, and $\varepsilon'=\varepsilon + (1/n)$ gives the desired bound on the alphabet size $q$.
\end{proof}

\section{The full lower bound: all $L$ and (ordinary) list-decoding.}
\label{sec:proof}

We now present the full proof of our main result, Theorem~\ref{thm:main}.
To do so, we need to combine the ideas in the Warmups 1, 2, and 3, and add some additional ideas.

\subsection{Additional Ingredients}

First, we need to generalize the warmups from average-radius list-decoding to (ordinary) list-decoding. To do so, we use similar ideas to \cite{BDG22}, but distill those ideas down to their combinatorial essence.
The idea is to choose our bad list-decoding configuration by first choosing a bad average-radius-list-decoding configuration $y,c_0,c_1,\dots,c_L$, as in the warmups. However, because the list-decoding center $y$ was much closer to $c_0$ than to each of $c_1, \ldots , c_L$, we will instead balance out the distances between them by ``transferring" agreements between $y$ and $c_0$ (of which there are almost $n$) to agreements between $y$ and $c_1,\dots,c_L$, until the $y$ has a similar number of agreements with each of $c_0, c_1, \dots,c_L$.
Specifically, the parts where we will transfer agreements from $c_0$ to $c_1, \ldots , c_L$ will be the intervals $I_1, \ldots I_L$ that we define in the proof: for $p=\frac{L}{L+1}(1-R-\varepsilon)$, we set aside the first $pn$ coordinates for the intervals $I_0, I_1, \ldots , I_L$, meaning that the agreement sets $A_1, \ldots , A_L$ will have to be subsets of $\{pn+1, \ldots , n\}$.

Next, we need to remove the distance requirement for (ordinary) list-decoding, which is more difficult than removing the distance requirement for average-radius-list-decoding in Section~\ref{ssec:warmup-3}.
To do so, we need a lemma similar to Lemma~\ref{lem:distance-0} that shows that a $(p,L)$-list-decodable code has a subcode with large distance and essentially the same rate. 
Clearly, $(p,L)$-list-decoding implies that every codeword is within distance $pn$ of at most $L$ codewords, so we can essentially assume our code has distance $pn$.
For technical reasons, this is not good enough.
In Lemma~\ref{lem:distance-1}, we show, more strongly, that a $(p,L)$-list-decodable code has a large subcode with distance $(p+\frac{p^L}{2L})n$.
Therefore, to show Theorem~\ref{thm:main}, it suffices for us to show it with the additional assumption that the minimum distance is $(p+\frac{p^L}{2L})n$, which is what we show in Theorem~\ref{thm:main-1}. 

To prove Theorem~\ref{thm:main-1}, we need additional ideas. In the warmup arguments, we needed to assume near-optimal minimum distance (namely $(1-R-O(\varepsilon))n$), to show that the codewords $c_1,\dots,c_L$ we find via the pigeonhole argument are pairwise distinct. 
To accommodate our weaker assumption of distance $p+\frac{p^L}{2L}$, we note that our pigeonholing actually gives substantially more than $L$ codewords, and in particular, certainly at least $L\cdot W$ codewords for some constant $W$.
Then, it is enough to show that our pigeonholing never produces $W$ equal codewords (rather than two equal codewords).
Here, a relaxed condition on our set system $\mathcal{F}$ suffices, namely that the sets have very large $W$-wise, rather than pairwise, unions (see Lemma~\ref{lem:sets} below).

We now present the lemmas described above. The first one implies that any $(p,L)$-list-decodable code has a large subcode with distance $(p+\frac{p^L}{2L})n$. 
A similar lemma appears in \cite[Theorem 15]{GN14}, and a lemma simliar in spirit appears in \cite[Theorem 6.1]{GST22}. 
We defer the proof to Appendix~\ref{app:distance-1}. 

\begin{lemma}
  Let $p\in(0,1)$.
  In any $(p,L)$-list-decodable code $C$, every codeword is within relative distance $\alpha := p+\frac{p^L}{2L}$ of at most $L'= O(L^2/p)$ codewords.
  Consequently, $C$ has a subcode of size at least $|C|/(L'+1)$ that has relative distance at least $\alpha$. 
  \label{lem:distance-1}
\end{lemma}

The next lemma says that, for sufficiently large $W$, we can choose a large family of sets with very large $W$-wise unions. The proof is a straightforward probabilistic argument, and we defer it to Appendix~\ref{app:sets}.

\begin{lemma}
  For all $1>\beta>\alpha>0$, for all positive integers $m$, there exists a constant $W = O\left(\log(1-\beta)/\log(1-\alpha)\right)$ and a family $\mathcal{F}$ of $2^{\Omega(m(1-\beta) \log(1-\alpha)/\log(1-\beta))}$ subsets of $[m]$, each of size $\alpha m$, such that all $W$-wise unions of subsets are of size at least $\beta m$.
  \label{lem:sets}
\end{lemma}

\subsection{Proof}

\begin{figure}
  \begin{center}
    \renewcommand\baselinestretch{1.4}
    \newcommand\dashvert[2]{\draw[dashed] (#1,-1.5) -- (#1,#2);}
    \newcommand\tick[1]{\draw (#1,-1.6) -- (#1,-1.4); \dashvert{#1}{1.5}} 
    \begin{tikzpicture}[scale=0.8]
        \def\ep{0.5}
        \def\ii{3}
        \def\ca{gray}
        \def\cb{yellow!40}
        \def\cc{orange!40}
        \def\cd{cyan}
        \def\ce{red!40}

        \draw[fill=\ce] (-\ep,2) rectangle (16,2.5);
        \draw[fill=\ca] (-\ep,1) rectangle (16,1.5);
        \draw[fill=\ca] (-\ep,0) rectangle (16,0.5);
        \draw[fill=\ce] (-\ep,-1) rectangle (16,-.5);
        
        \draw[fill=\cd] (-\ep,1) rectangle (0,1.5);
        \draw[fill=\cd] (-\ep,0) rectangle (0,.5);
        \draw[fill=\cd] (-\ep,-1) rectangle (0,-.5);
        \draw[fill=\cc] (0*\ii,1) rectangle (1*\ii,1.5);
        \draw[fill=\cc] (0,-1) rectangle (\ii,-.5);
        \draw[fill=\cb] (1*\ii,0) rectangle (2*\ii,.5);
        \draw[fill=\cb] (\ii,-1) rectangle (2*\ii,-.5);
        \draw[fill=\ce] (2.5*\ii,1) rectangle (\ep+3.5*\ii,1.5);
        \draw[fill=\ce] (3.3*\ii,0) rectangle (\ep+4.3*\ii,0.5);

        \tick{-\ep}
        \tick{0}
        \tick{\ii}
        \tick{2*\ii}
        \tick{2.5*\ii}
        \tick{\ep+3.5*\ii}
        \tick{-\ep}
       n  \tick{2*\ii}
        \draw (-\ep,-1.5) -- (16,-1.5);
        \draw[line width=3pt] (-\ep,-1) rectangle (16,-.5);
         \dashvert{-\ep}{2.7}
         \dashvert{0}{2.7}
         \dashvert{\ii}{2.7}
         \dashvert{2*\ii}{2.7}
        
        \node at (-2*\ep,2.25) {$c_0$};
        \node at (-2*\ep,1.25) {$c_1$};
        \node at (-2*\ep,0.25) {$c_2$};
        \node at (-2*\ep,-0.75) {$y$};
        \node at (-0.5*\ep,3) {$I_0$};
        \node at (0.5*\ii,3) {$I_1$};
        \node at (1.5*\ii,3) {$I_2$};
        \node at (0.5*\ep+3*\ii,1.25) {$A_1$};
        \node at (0.5*\ep+3.8*\ii,0.25) {$A_2$};
        \node[align=center] at (-0.5*\ep,-2.5) {$d_0$\\$\frac{8\varepsilon n}{3}$};
        \node[align=center] at (0.5*\ii,-2.5) {$d_1$\\$\frac{(1-R-5\varepsilon) n}{3}$};
        \node[align=center] at (1.5*\ii,-2.5) {$d_1$\\$\frac{(1-R-5\varepsilon) n}{3}$};
        \node[align=center] at (3*\ii+0.5*\ep,-2.5) {$a_\mathcal{F}$\\$k-1$};
        \node at (\ii-0.5*\ep,4.2) {$I_*$};
        \node at (\ii-0.5*\ep,3.6) {\scalebox{15}[3]{\rotatebox{270}{$\{$}}};
  \end{tikzpicture}
  \end{center}
  \caption{The agreement pattern we search for via pigeonhole in our upper bound, for $L=2$. Codeword $c_0$ differs from $y$ in at most $d_0+2d_1$ places and codewords $c_1$ and $c_2$ differ from $y$ in at most $n-d_0-d_1-a_\mathcal{F}$ places.} 
  \label{fig:agreement}
\end{figure}
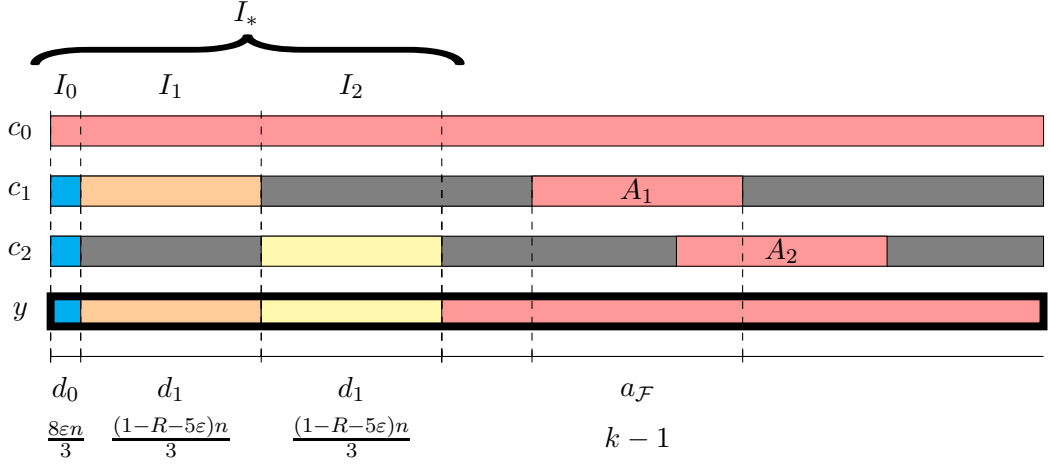

Now, we put the above ideas all together, and generalize to all $L$, to give the full theorem.
By Lemma~\ref{lem:distance-1}, the following theorem implies Theorem~\ref{thm:main}
\begin{theorem}
  Let $L\ge2$ be a fixed constant and $R\in(0,1)$.
  There exists $n_0=n_0(L,R)$ such that the following holds for all $n\ge n_0$ and $\varepsilon\ge 1/n$.
  \label{thm:main-1}
  Let $C$ be a code of length $n$ with alphabet size $q$ that is $(p,L)$-list-decodable for $p=\frac{L}{L+1}(1-R-\varepsilon)$. Suppose also that $C$ has minimum distance at least $(p + \frac{p^L}{2L})n$. Then $q \ge 2^{\Omega_{L,R}(1/\varepsilon)}$.
\end{theorem}
\begin{proof}
  Since we allow $\varepsilon\ge 1/n$, we need to deal with rounding issues. 
  It suffices to consider the case where $R$ and $\varepsilon$ are both integer multiples of $\frac{L+1}{n}$: if the Proposition is true in these cases with $q\ge 2^{\alpha_{L,R}'/\varepsilon}$, and we are given a code $C$ of rate $R$ and satisfying the Proposition's conditions for some $\varepsilon$, we can round $R$ down to the nearest multiple of $\frac{L+1}{n}$, call it $R'$,  and round $\varepsilon+\frac{L+1}{n}$ up to the next multiple of $\frac{L+1}{n}$, call it $\varepsilon'$. 
  Taking $C'$ to be any subcode of $C$ of rate $R'$, we can check $C'$ satisfies the requirements of the Proposition (as $1-R-\varepsilon > 1-R'-\varepsilon'$), so $q\ge 2^{\alpha'_{R,L}/\varepsilon'}\ge 2^{\alpha_{R,L}/\varepsilon}$ where $\alpha_{R,L} = \frac{\alpha'_{R,L}}{40(L+1)}$ and where the latter inequality uses $\varepsilon\ge 1/n$.
  Thus, we assume for the rest of the proof $R$ and $\varepsilon$ are multiples of $\frac{L+1}{n}$ (so that everything that ``should'' be an integer is an integer).

  Since we suppress factors of $R$ and $L$ in the overall alphabet size bound, it suffices to consider $\varepsilon$ sufficiently small compared to $R,L$.
  With hindsight, define the following integer parameters
  \begin{align}
    a_\mathcal{F} &\coloneqq k-1, \label{eq:parameter-1} \\
    a_\cup &\coloneqq \left\lfloor\left(1 - p - \frac{p^L}{4L} \right)n\right\rfloor > \left( 1-p-\frac{p^L}{2L} \right)n, \label{eq:parameter-2} \\
    d_1 &\coloneqq \left(\frac{1-R-5\varepsilon}{L+1}\right)n, \label{eq:parameter-3} \\
    d_0 &\coloneqq \left(\frac{4L\varepsilon}{L+1}\right)n, \label{eq:parameter-4}
  \end{align}
  By assumption that $R$ and $\varepsilon$ are multiples of $\frac{L+1}{n}$, the parameters $d_0$ and $d_1$ are integers.
  We remark that $d_0$ and $d_1$ are chosen to satisfy the following equation and inequalities specifically:
  \begin{align}
    d_0+Ld_1 &= pn, \label{eq:constraint-1} \\
    n - d_0 - d_1 - a_\mathcal{F} &\le pn, \label{eq:constraint-2} \\
    d_0 &\le 4\varepsilon n. \label{eq:constraint-3}
  \end{align}
  Now, let $I_0,\dots,I_{L}$ be consecutive subintervals of $[n]$ (in that order), such that interval $I_0=\{1,\dots,d_0\}$ has size $d_0$, and intervals $I_1,\dots,I_{L}$ have size $d_1$.
  Define $I_* \coloneqq I_0 \cup I_1 \cup \cdots \cup I_L$. Note that $\abs{I_*} = d_0+Ld_1 =^{\eqref{eq:constraint-1}} pn$.

  We have that, for large enough $n\ge n_0(L,R)$ and small enough $\varepsilon \le \varepsilon_0(L,R)$,
  \begin{equation}
    1 - \frac{1}{4L}\left(\frac{1-R}{4}\right)^{L} \ge 1 - \frac{p^L}{4L(1-p)} \ge \frac{a_\cup}{n-pn} \ge \frac{a_\mathcal{F}}{n-pn} \ge \frac{(L+1)R}{1+(L+1)R} \ .
  \label{eq:simple-ineq}
  \end{equation}
  The first inequality holds as $p=\frac{L}{L+1}(1-R-\varepsilon)\ge \frac{1-R}{4}$ and $1-p \le 1$.
  The second inequality follows from definition of $a_\cup$.
  The third inequality holds because $a_\cup>(1-p-\frac{p^L}{2L})n>(1-\delta)n$, where $\delta$ is the relative distance, and $(1-\delta)n\ge k-1= a_\mathcal{F}$ by the Singleton bound.
  The fourth inequality holds because $\frac{a_\mathcal{F}}{n-pn} = \frac{R-1/n}{\frac{1}{L+1} + \frac{L}{L+1}(R+\varepsilon)}\ge \frac{(L+1)R}{1+(L+1)R}$ for $\varepsilon\le \varepsilon_0(L,R)$ sufficiently small and $n\ge n_0(L,R)$ sufficiently large.

  That is, if we define $\alpha \triangleq a_\mathcal{F}/(n-pn)$ and $\beta \triangleq  a_\cup/(n-pn)$, then \eqref{eq:simple-ineq} implies that $\alpha$ and $\beta$ lie in an interval contained in $(0,1)$ that is completely determined by $R$ and $L$ and not on $\varepsilon$.
  By applying Lemma~\ref{lem:sets} on the ground set $[n] \setminus I_*$, we get a constant $W=O_{L,R}(1)$ and a family $\mathcal{F}$ of $2^{\Omega_{L,R}(n)}$ 
  subsets of $[n]\setminus I_*$, each of size $a_\mathcal{F}$, such that every $W$-wise union of sets in $\mathcal{F}$ has size at least $a_\cup$. Here, \eqref{eq:simple-ineq} ensures that $W$ and the number of subsets $2^{\Omega_{L,R}(n)}$ do not depend on $\varepsilon$. 

  Consider picking a uniformly random codeword $c \in C$. 
  For each $A \in \mathcal{F}$, let $\mathscr{E}_A$ be the event that another codeword $c'$ agrees with $c$ on $A$, i.e., $c|_A = c'|_A$. For any $A \in \mathcal{F}$, there at most $q^{a_\mathcal{F}}$ possible values of $c|_A$, and thus at most $q^{a_\mathcal{F}}$ codewords $c$ for which $c|_A$ uniquely determines $c$. Hence,
  \begin{equation}
      \Pr[\neg\mathscr{E}_A] < \frac{q^{a_\mathcal{F}}}{q^k} \overset{\eqref{eq:parameter-1}}{=} \frac{1}{q} \ .
      \label{eq:uniqueness-ub-2}
  \end{equation}
   For each codeword $c$, define the set $\mathcal{F}_c \coloneqq \{A \in \mathcal{F} : \text{$\mathscr{E}_A$ occurs}\}$. For each $A \in \mathcal{F}_c$, we can find, by definition, a codeword $f^A(c) \in C \setminus \{c\}$ such that $f^A(c)|_A = c|_A$. By linearity of expectation and \eqref{eq:uniqueness-ub-2}, we thus find that
  \begin{equation}
      \E\left[\abs{\mathcal{F}_c}\right] = \E\left[\sum_{A \in \mathcal{F}}{\1\{\mathscr{E}_A\}}\right] > \sum_{A \in \mathcal{F}}{\left(1 - \frac{1}{q}\right)} \ge \frac{\abs{\mathcal{F}}}{2} \ .
  \end{equation}
  Hence we can find a codeword $c \in C$ for which $\abs{\mathcal{F}_c} > \abs{\mathcal{F}}/2$. Fix this codeword $c_0=c$. To prove Theorem~\ref{thm:main-1}, it suffices to prove that $2\cdot W\cdot L\cdot q^{d_0} \ge |\mathcal{F}|$.  Suppose for contradiction that
  \begin{align}
    W\cdot L\cdot q^{d_0} < |\mathcal{F}|/2
  \label{eq:false}
  \end{align}
  By pigeonhole and \eqref{eq:false}, there are $WL$ sets $A_1,\dots,A_{WL} \in \mathcal{B}_{c_0}$ such that $f^{A_1}(c_0),\dots,f^{A_{WL}}(c_0)$ agree on the coordinates in $I_0$.
  Further, no $W$ of the codewords can be equal: 
  if, for example, $f^{A_1}(c_0)=f^{A_2}(c_0)=\cdots=f^{A_W}(c_0)$, then this codeword $f^{A_1}(c_0)$ agrees with $c_0$ on the coordinates $ \cup_{i=1}^W{A_i}$, which by construction of $\mathcal{F}$ has size at least $a_\cup >^{\eqref{eq:parameter-2}} (1-p-p^L/2L)n$.
  Thus, we have found two distinct codewords that disagree on less than $(p+p^L/2L)n$ positions, which contradicts the minimum distance assumption of our code.
  Thus, no $W$ codewords among $f^{A_1}(c_0),\dots,f^{A_{WL}}(c_0)$ are equal. In particular, this implies that we have at least $L$ distinct codewords.
  Without loss of generality, say the codewords $c_1\defeq f^{A_1}(c_0),c_2\defeq f^{A_2}(c_0),\dots,c_L\defeq f^{A_L}(c_0)$ are pairwise distinct. 

  Let $y\in[q]^n$ be the list-decoding center that agrees with $c_1$ on coordinates $I_0$ (and thus $c_2, \ldots , c_L$), agrees with $c_j$ on coordinates $I_j$ for $j=1,\dots,L$, and agrees with $c_0$ elsewhere (see Figure~\ref{fig:agreement}). Let us analyze the distance of $y$ to the $L+1$ codewords $c_0, c_1, \ldots , c_L$ in two cases:
  \begin{enumerate}
    \item  First, by construction of $y$, the codeword $c_0$ can only disagree with $y$ on $I_*$. Thus the distance between $y$ and $c_0$ is most $\abs{I_*} = m$, which is at most $pn$ by \eqref{eq:constraint-1}.
    \item For $j=1,\dots,L$, by construction of $y$, codeword $c_j$ agrees with $y$ on $I_0$, $I_j$, and $A_j$. Thus, $c_j$ disagrees with $y$ on at most $n - \abs{I_0} - \abs{I_j} - \abs{A_j} = n - d_0 - d_1 - a_\mathcal{F}$ coordinates, which is at most $pn$ by \eqref{eq:constraint-2}. 
  \end{enumerate}
  Thus, we have found $L+1$ distinct codewords each with Hamming distance at most $pn$ from $y$, contradicting that $C$ is $(p,L)$-list-decodable.
  Hence, \eqref{eq:false} is false, giving us our desired lower bound $q\ge (\frac{|\mathcal{F}|}{2WL})^{1/d_0} \ge^{\eqref{eq:constraint-3}} 2^{\Omega_R(1/\varepsilon)}$.
\end{proof}

\begin{proof}[Proof of Theorem~\ref{thm:main}]
  Let $C$ be a code that is $(p,L)$-list-decodable for $p=\frac{L}{L+1}(1-R-\varepsilon)$.
  By Lemma~\ref{lem:distance-1} and for $\varepsilon$ sufficiently small, every codeword $C$ has a subcode $C'$ of $C$ with minimum distance $(p+\frac{p^L}{2L})n$ and rate $R' = R - \frac{\log (L'+1)}{n}=R'-o(1)$.
Apply Theorem~\ref{thm:main-1} to the subcode $C'$ with rate $R'$ and $\varepsilon'=\varepsilon + \frac{\log(L'+1)}{n} \le \varepsilon+o(1)$, and use that $n$ is sufficiently large to obtain the result. 
\end{proof}

\section{Concluding Remarks}

As alluded to in Remark~\ref{remark:growing-l}, Theorem~\ref{thm:main} focuses on the case when $L$ is a fixed constant independent of $1/\eps$. It nonetheless leaves open the question of showing an alphabet size lower bound for $L$ growing with $1/\eps$. 
\begin{question}
    Can we show that all codes (for sufficiently large $n$) that are $(\frac{L}{L+1}(1-R-\varepsilon),L)$-list-decodable require alphabet size $q\ge 2^{\Omega_R(1/\varepsilon)}$ (independent of $L$)?
\end{question}
In the most general case, our current methods give a constant $\alpha_{L,R}$ in Theorem~\ref{thm:main} that is at most $\exp(-O_R(L))$. However, in special cases, we get better bounds.

For average-radius list-decoding, such an $L$-independent alphabet size lower bound of $q\ge 2^{\Omega_R(1/\varepsilon)}$ follows from the warmup arguments in Sections~\ref{ssec:warmup-1},~\ref{ssec:warmup-2}, and~\ref{ssec:warmup-3}.
In particular, one can check that Proposition~\ref{pr:main-1} holds if we assume $(\frac{L}{L+1}(1-R-\varepsilon),L)$-average-radius-list-decoding, again with minimum distance $(1-R-\varepsilon)n$, and we again get an alphabet size lower bound of $q\ge 2^{\alpha_R/\varepsilon}$, independent of $L$ (the lower bound will in fact be $\binom{n}{Rn}^{1-\delta}$ for some $\delta=\delta(\varepsilon)\to0$ decreasing with $\varepsilon\to0$.).\footnote{More details: We still take $|I_0|=4\varepsilon, \alpha=R-\varepsilon, \beta=R+\varepsilon$, and $\mathcal{F}$ to be the same family of subsets. We prove $q^{|I_0|} \ge |\mathcal{F}|/L$, and, assuming not (for the sake of contradiction), we find a codeword $c$ and $L$ codewords $f^{A_1}(c),\dots,f^{A_L}(c)$ agreeing with each other on $I_0$ and where $f^{A_i}(c)$ agrees with $c$ on $A_i$. As $|A_i\cup A_j|\ge (R+\varepsilon)n$, these codewords are pairwise distinct or else we contradict the code distance. Now we choose the list-decoding center $y$ in the same way, and it has total distance $|I_0| + L\cdot(n-|I_0| - \alpha n)  = 4\varepsilon n + L(1-R-3\varepsilon)n < L(1-R-\varepsilon)n$ to these codewords, contradiction.}
This implies that Proposition~\ref{pr:main-2} also generalizes to give an optimal alphabet size lower bound of $q\ge 2^{\Omega_R(1/\varepsilon)}$ for $(\frac{L}{L+1}(1-R-\varepsilon),L)$-average-radius-list-decoding (without an additional distance assumption).

For ordinary list-decoding, if we can assume a fixed relative distance $\delta > p$ depending only on $R$, then $\alpha_{L,R}$ can be improved to $\Omega_R(1/L)$ by following the same argument as in Theorem~\ref{thm:main-1}.
Combining this bound when $L\le 1/\sqrt{\varepsilon}$ with the list-decoding capacity theorem (Remark~\ref{remark:growing-l}) when $L\ge 1/\sqrt{\varepsilon}$ gives an $L$-independent alphabet size lower bound of $q \ge 2^{\Omega_R(1/\sqrt{\varepsilon})}$.

\bibliographystyle{alpha}
\bibliography{bib}

\begin{thebibliography}{BGM23}

\bibitem[AGL23]{AGL23}
Omar Alrabiah, Venkatesan Guruswami, and Ray Li.
\newblock Randomly punctured {R}eed--{S}olomon codes achieve list-decoding
  capacity over linear-sized fields.
\newblock {\em arXiv preprint arXiv:2304.09445}, 2023.

\bibitem[AS16]{AS16}
Noga Alon and Joel~H Spencer.
\newblock {\em The probabilistic method}.
\newblock John Wiley \& Sons, 2016.

\bibitem[BDG22]{BDG22}
Joshua Brakensiek, Manik Dhar, and Sivakanth Gopi.
\newblock Improved field size bounds for higher order {MDS} codes.
\newblock {\em arXiv preprint arXiv:2212.11262}, 2022.

\bibitem[BGM22]{BGM22}
Joshua Brakensiek, Sivakanth Gopi, and Visu Makam.
\newblock Lower bounds for maximally recoverable tensor codes and higher order
  {MDS} codes.
\newblock {\em IEEE Transactions on Information Theory}, 68(11):7125--7140,
  2022.

\bibitem[BGM23]{BGM23}
Joshua Brakensiek, Sivakanth Gopi, and Visu Makam.
\newblock Generic {Reed-Solomon} codes achieve list-decoding capacity.
\newblock In {\em Proceedings of the 55th Annual {ACM} Symposium on Theory of
  Computing (STOC)}, pages 1488--1501, 2023.

\bibitem[GN14]{GN14}
Venkatesan Guruswami and Srivatsan Narayanan.
\newblock Combinatorial limitations of average-radius list-decoding.
\newblock {\em IEEE Transactions on Information Theory}, 60(10):5827--5842,
  2014.

\bibitem[GS95]{GS95}
Arnaldo Garcia and Henning Stichtenoth.
\newblock A tower of {A}rtin-{S}chreier extensions of function fields attaining
  the {D}rinfeld-{V}l\u{a}dut bound.
\newblock {\em Inventiones mathematicae}, 121(1):211--222, 1995.

\bibitem[GST22]{GST22}
Eitan Goldberg, Chong Shangguan, and Itzhak Tamo.
\newblock Singleton-type bounds for list-decoding and list-recovery, and
  related results.
\newblock In {\em 2022 IEEE International Symposium on Information Theory
  (ISIT)}, pages 2565--2570. IEEE, 2022.

\bibitem[GZ23]{GZ23}
Zeyu Guo and Zihan Zhang.
\newblock Randomly punctured {Reed-Solomon} codes achieve the list decoding
  capacity over polynomial-size alphabets.
\newblock In {\em FOCS 2023, to appear, arXiv preprint arXiv:2304.01403}, 2023.

\bibitem[Plo60]{Plotkin60}
Morris Plotkin.
\newblock Binary codes with specified minimum distance.
\newblock {\em IRE Transactions on Information Theory}, 6(4):445--450, 1960.

\bibitem[Rot22]{Roth22}
Ron~M Roth.
\newblock Higher-order {MDS} codes.
\newblock {\em IEEE Transactions on Information Theory}, 68(12):7798--7816,
  2022.

\bibitem[{Sin}64]{Singleton64}
Richard {Singleton}.
\newblock Maximum distance $q$-nary codes.
\newblock {\em IEEE Trans. Inform. Theory}, 10(2):116--118, April 1964.

\bibitem[ST20]{ST20}
Chong Shangguan and Itzhak Tamo.
\newblock Combinatorial list-decoding of {R}eed-{S}olomon codes beyond the
  {Johnson} radius.
\newblock In {\em Proceedings of the 52nd Annual ACM Symposium on Theory of
  Computing}, STOC 2020, pages 538--551, 2020.

\bibitem[TVZ82]{TVZ82}
Michael~A Tsfasman, SG~Vl{\u{a}}dut, and Th~Zink.
\newblock Modular curves, {S}himura curves, and {G}oppa codes, better than
  {V}arshamov-{G}ilbert bound.
\newblock {\em Mathematische Nachrichten}, 109(1):21--28, 1982.

\end{thebibliography}

\appendix

\section{Deferred Proofs}

\subsection{Proof of Proposition~\ref{pr:codes-exist}}
\label{sec:codes-exist}

In this appendix, we prove Proposition~\ref{pr:codes-exist}. We remark that the proof we present seamlessly extends to the notion of average-radius-list-decodabililty.

\begin{proof}[Proof of Proposition~\ref{pr:codes-exist}]
Fix an alphabet $\Sigma$ of size $q$, and set $N \coloneqq \lfloor q^{Rn} \rfloor$ and $p \coloneqq \frac{L}{L+1}(1-R-\eps)$. Consider a code $C = \{c^{(1)}, \ldots , c^{(N)}\} \subseteq \Sigma^n$ where each $c^{(\ell)}$ is chosen independently and uniformly at random from $\Sigma^n$. Set $p \coloneqq \frac{L}{L+1}(1-R-\eps)$. For any set $I \subseteq [N]$ of size $L+1$ and word $w \in \Sigma^n$, let $\mathscr{B}^w_I$ be the event that $d(c^{(\ell)}, w) \le p n$ for all $\ell \in I$. Define the sets $A_\ell \coloneqq \{i \in [n] : c^{(\ell)}_i = w_i\}$ and $E_i \coloneqq \{\ell \in I : c^{(\ell)}_i = w_i\}$. Then by double counting, we find that
\begin{equation}
\label{eq:edge-lb}
\sum_{i=1}^n{\abs{E_i}} = \sum_{\ell \in I}{\abs{A_\ell}} = \sum_{\ell \in I}{(n - d(c^{(\ell)}, w))} \ge (1 + LR + L\eps)n \ .
\end{equation}
Consider the hypergraph $\mathcal{H}$ with vertices $V(\mathcal{H}) = I$ and hyperedges $E(\mathcal{H}) = \{E_1, \ldots , E_n\})$. Let $\mathscr{X}^{\mathcal{H}}_I$ be the event that $c^{(\ell)}_i = w_i$ for all $\ell \in E_i$ and $i \in [n]$. Since each $c^{(\ell)}$ is independently and uniformly chosen from $\Sigma^n$, each hyperedge $E_i$ 'imposes' $\abs{E_i}$ constraints. Thus by using Inequality~\eqref{eq:edge-lb} and union bounding over all choices of $\mathcal{H}$, we find that
\begin{align}
\label{eq:bad-event-ub}
\Pr[\mathscr{B}^w_i] &\le \Pr\left[\text{$\exists$ hypergraph $\mathcal{H}$ such that $\mathscr{X}^{\mathcal{H}}_I$ occurs}\right] \nonumber \\
&\le 2^{(L+1)n} q^{-\sum_{i=1}^n{\abs{E_i}}} \nonumber \\
&\le 2^{(L+1)n} q^{-(1 + LR + L\eps)n} \ .
\end{align}
Now, pick $q \ge 2^{3/\eps}$. Using Inequality~\eqref{eq:bad-event-ub}, we conclude that
\begin{align*}
\E\biggl[\sum_{\substack{w \in \Sigma^n \\ I \subseteq [N], \abs{I} = L+1}}{\1\{\mathscr{B}^w_I\}}\biggr] &\le q^n \cdot q^{(L+1)Rn} \cdot 2^{(L+1)n} q^{-(1 + LR + L\eps)n} \\
&= q^{Rn} \cdot 2^{(L+1)n} \cdot q^{-L\eps n} \\
&\le q^{Rn} \cdot 2^{-Ln} \ .
\end{align*}
Thus we can find a code $C$ of rate $R$ such that $\mathscr{B}^w_I$ occurs for at most $q^{Rn} \cdot 2^{-Ln}$ choices of $w$ and $I$. For each such pair $(w,I)$, fix an index $i_{w, I}  \in I$ and consider the expurgated code $C' \coloneqq C \setminus C_b$. Then $\abs{C'} \ge q^{Rn}(1-2^{-Ln}) = q^{n(R-o(1))}$. Furthermore, since we removed all the 'bad' codewords from, none of the events $\mathscr{B}^w_I$ now occur in $C'$, implying that $C'$ is a $\left(\frac{L}{L+1}(1-R-\eps), L\right)$-list-decodable code.
\end{proof}

\subsection{Proof of Lemma~\ref{lem:distance-1}}
\label{app:distance-1}
\begin{proof}[Proof of Lemma~\ref{lem:distance-1}]
  Let $M=\ceil{\frac{L^2}{p}}$.
  We may assume without loss of generality that the all-0s string $0^n$ is in the code, and by symmetry it suffices to show that there are at most $L+M-1$ codewords within distance at most $\alpha n$ from $0^n$, i.e., (Hamming) weight at most $\alpha n$.
  
  There are clearly at most $L$ codewords of weight at most $pn$ be the list-decoding property (centered at 0). 
  Suppose for contradiction there are $M$ codewords $c_1,\dots,c_M$ of weight between $pn$ and $\alpha n$.

  We claim there are $L$ nonzero codewords $c_{j_1},\dots,c_{j_L}$ such that
  \begin{align}
    \abs{\supp(c_{j_1})\cap \cdots\cap\supp(c_{j_L})} \ge p^Ln.
  \end{align}
  For $i=1,\dots,n$, let $a_i$ denote the number of $j\in[M]$ such that $i\in \supp(c_j)$.
  Let $T$ denote the number of tuples $(i,j_1,\dots,j_L)\in[n]\times [M]^L$ with $j_1<j_2<\cdots<j_L$ such that $i\in \supp(c_{j_1})\cap \cdots\cap\supp(c_{j_L})$.
  By double counting,
  \begin{align}
    \binom{M}{L} \max_{j_1<\cdots<j_L}\abs{\supp(c_{j_1})\cap \cdots\cap\supp(c_{j_L})}
    \ge T
    = \sum_{i=1}^{n} \binom{a_i}{L}
    \ge n\cdot \binom{p M}{L}
    \label{}
  \end{align}
  The last inequality uses convexity of $\binom{\cdot}{L}$ and that $\sum_{i=1}^n a_i =\sum_{j=1}^{M} |\supp(c_j)| \ge M\cdot p n$.
  Rearranging, we have
  \begin{align}
    \max_{j_1<\cdots<j_L}\abs{\supp(c_{j_1})\cap \cdots\cap\supp(c_{j_L})}
    &\ge n\cdot \frac{\binom{p M}{L}}{\binom{M}{L}} \nonumber\\
    &\ge n\cdot \frac{(p M)(p M-1)\cdots(p M-L+1)}{M^L} \nonumber\\
    &\ge n\cdot p^{L} \cdot \left(1-\frac{\binom{L}{2}}{p M}\right) \nonumber\\
    &\ge n\cdot \frac{p^L}{2}.
    \label{eq:distance-1}
  \end{align}
  by the bound on $M$.
  The third inequality uses that $(1-a_1)(1-a_2)\cdots(1-a_n)\ge 1-(a_1+\cdots+a_n)$.

  Without loss of generality, \eqref{eq:distance-1} is realized by
  \begin{align}
    |\supp(c_1)\cap\supp(c_2)\cap\cdots\cap \supp(c_L)|\ge \frac{p^L}{2} n.
    \label{eq:distance-2}
  \end{align}
  Now consider the codewords $0,c_1,c_2,\dots,c_L$, and let $S_1,S_2,\dots,S_L$ be pairwise disjoint subsets of $\supp(c_1)\cap\cdots\cap \supp(c_L)$ of size $(\alpha-p)n=\frac{p^L}{2L}n$.
  There sets exist because of \eqref{eq:distance-2}.
  Consider the word $w$ such that $w$ agrees with $c_j$ on $S_j$ for $j=1,\dots,L$, and is 0 otherwise.
  Note that the distance from $w$ to 0 is at most $|S_1\cup S_2\cdots\cup S_j| \le \frac{p^L}{2} n < pn$. 
  The distance from $w$ to $c_1$ is at most
  \begin{align}
    |\supp(c_1)\cup S_2\cup S_3\cup\cdots\cup S_L \setminus S_1| = |\supp(c_1)\setminus S_1| = |\supp(c_1)|-|S_1| \le pn. 
  \end{align}
  Thus, we've found $L+1$ codewords $0,c_1,c_2,\dots,c_L$ within distance $pn$ of $w$, contradicting list-decodability.
  Hence, there are at most $M-1$ codewords with weight between $pn$ and $\alpha n$, so there are at most $L+M-1$ codewords with weight at most $\alpha n$, as desired.

  The subcode of $C$ with minimum distance $\alpha n$ can be chosen greedily from $C$.
\end{proof}

\subsection{Proof of Lemma~\ref{lem:sets}}
\label{app:sets}
\begin{proof}[Proof of Lemma~\ref{lem:sets}]
  With hindsight, let $W$ be the positive integer such that $(1-\beta)/2 > (1-\alpha)^W \ge (1-\beta)(1-\alpha)/2$, and let $M=2^{(1-\beta) m/6W}$.
  Pick $M$ sets independently by including each element of $[m]$ independently with probability $\alpha_0 =  \alpha - \frac{1}{m^{1/3}}$.
  By standard concentration arguments, with high probability, at most $o(M)$ of the sets are of size more than $\alpha m$.

  For large enough $m$, the probability we have some $W$-wise union of size less than $\beta m$ is at most
  \begin{align}
    \binom{M}{W}\cdot \Pr\left[\Binomial(1-(1-\alpha_0)^W,m) < \beta m\right]
    &= \binom{M}{W}\cdot \Pr\left[\Binomial\left((1-\alpha_0)^W,m\right) > (1-\beta)m\right] \nonumber \\
    &\le \binom{M}{W}\cdot \Pr\left[\Binomial\left((1-\beta)/2,m\right) > (1-\beta)m\right] \nonumber \\
    &\le M^W\cdot e^{-(1-\beta)m/6}  \ll 1.
    \label{}
  \end{align}
  where above used the Chernoff bound \eqref{eq:chernoff} with $\delta=1$.
  Thus, with high probability, all $W$-wise unions have size at least $\beta m$ and at least $M-o(M)$ sets are of size at most $\alpha m$.
  Hence, some choice of sets exists.
  Taking the $M-o(M)$ sets of size at most $\alpha m$, and appending arbitrary elements to them until they have size exactly $\alpha m$, gives our desired family.
\end{proof}

\end{document}